\documentclass[twoside,runningheads,envcountsame,envcountsect,oribibl,orivec]{llncs}

\usepackage{isabelle,isabellesym}
\usepackage{amsmath,amssymb}
\usepackage{times}
\usepackage{color}
\usepackage{mathbbol}
\def \unit {\mathbb{1}}
\def \zero {\mathbb{O}}

\newcommand{\sskip}{\mathbf{skip}}
\newcommand{\iif}{\mathbf{if}}
\newcommand{\tthen}{\mathbf{then}}
\newcommand{\eelse}{\mathbf{else}}
\newcommand{\ffi}{\mathbf{fi}}
\newcommand{\wwhile}{\mathbf{while}}
\newcommand{\ddo}{\mathbf{do}}
\newcommand{\ood}{\mathbf{od}}

\newcommand{\eemp}{\mathbf{emp}}

\newcommand{\triple}[3]{\{#1\}\ #2\ \{#3\}}
\newcommand{\ifstat}[3] {\iif \ #1 \ \tthen \ #2 \ \eelse \ #3 \ \ffi}
\newcommand{\whileloop}[2] {\wwhile \ #1 \ \ddo \ #2 \ \ood}

\newcommand{\predT}[1]{[#1]}

\newcommand{\wand}{-\hspace{-0.42em}\ast\ }

\newcommand{\emptylist}{[~]}
\newcommand{\mapslist}{\ [\mapsto]\ }

\newcommand{\spec}[2]{{\isasymlbrakk}#1, #2{\isasymrbrakk}}

\isabellestyle{it}

\title {Principles for Verification Tools: Separation Logic}

\author {
  	Brijesh Dongol
	\and
	Victor B. F. Gomes
	\and 
	Georg Struth 
}

\institute {
	Department of Computer Science, University of Sheffield \\
	\email{\{b.dongol,v.gomes,g.struth\}@shefield.ac.uk}
}

\begin{document}

\maketitle

\begin{abstract}
  A principled approach to the design of program verification and
  construction tools is applied to separation logic. The control flow
  is modelled by power series with convolution as separating
  conjunction. A generic construction lifts resource monoids to
  assertion and predicate transformer quantales. The data flow is
  captured by concrete store/heap models. These are linked to the
  separation algebra by soundness proofs. Verification conditions and
  transformation laws are derived by equational reasoning within the
  predicate transformer quantale.  This separation of concerns makes
  an implementation in the Isabelle/HOL proof assistant simple and
  highly automatic. The resulting tool is correct by construction; it
  is explained on the classical linked list reversal example.

\end{abstract}

%%%%%%%%%%%%%%%%%%%%%%%%%%%%%%%%%%%%%%%%%%%%%%%%%%%%%%%%%%%%%%%%%%%%%%%%%%%%%%

\section{Introduction}\label{S:intro}

Separation logic yields an approach to program verification that has
received considerable attention over the last decade. It is designed
for local reasoning about a system's states or resources, allowing one
to isolate the part of a system that an action affects from the
remainder, which is unaffected. This capability is provided by the
idiosyncratic separating conjunction operator together with the frame
inference rule, which makes local reasoning modular.  A key
application is program verification with pointer data
structures~\cite{Reynolds02, OHRY01}; but the method has also been
used for modular reasoning about concurrent
programs~\cite{OHearn07,Vafeiadis} or fractional
permissions~\cite{BornatCOP05}.

Separation logic is currently supported by a large number of tools; so
large that listing them is beyond the scope of this
paper. Implementations in higher-order interactive proof
assistants~\cite{Weber04,Tuerk,ChlipalaMMSW09,NICTA} are particularly
relevant to this article.  In comparison to automated tools or tools
for decidable fragments, they can express more program properties, but
are less effective for proof search. Ultimately, an integration seems
desirable.

This article adds to this tool chain (and presents another
implementation in the Isabelle/HOL theorem proving
environment~\cite{NPW02}).  However, our approach is different in
several respects. It focusses almost entirely on making the control
flow layer as simple as possible and on separating it cleanly from the
data flow layer.  This supports the integration of various data flow
models and modular reasoning about these two layers, with assignment
laws providing and interface.

To achieve this separation of concerns, we develop a novel algebraic
approach to separation logic, which aims to combine the simplicity of
original logical approaches~\cite{OHRY01} with the abstractness and
elegance of O'Hearn and Pym's categorical logic of bunched
implications~\cite{OHearnP99} in a way suitable for formalisation in
Isabelle. Our approach is based on \emph{power series}~\cite{DHS14},
as they have found applications in formal language and automata
theory~\cite{BerstelReutenauer,Handbook}. Their use in the context of
separation logic is  a contribution in itself.

In a nutshell, a \emph{power series} is a function $f:M\to Q$ from a
partial monoid $M$ into a quantale $Q$. Defining addition of power
series by lifting addition pointwise from the quantale, and
multiplication as \emph{convolution}
\begin{equation*}
  (f \otimes g)\ x = \sum_{x=y\circ z} f\ y \odot g\ z,
\end{equation*}
it turns out that the function space $Q^M$ of power series forms
itself a quantale~\cite{DHS14}. In the particular case that $M$ is
commutative (a \emph{resource monoid}) and $Q$ formed by the booleans
$\mathbb{B}$ (with $\odot$ as meet), one can interpret power series
as assertions or predicates over $M$.  Separating conjunction then
arises as a special case of convolution, and, in fact, as a language
product over resources. The function space $\mathbb{B}^M$ is the
assertion quantale of separation logic. The approach generalises to
power series over program states modelled by store-heap pairs. This
generalisation is needed for applications.

Using lifting results for power series again, we construct the
quantale-like algebraic semantics of predicate transformers over
assertion quantales. We characterise the monotone predicate
transformers and derive the inference rules of Hoare logic (without
assignment) in generic fashion within this subalgebra. Furthermore, we
derive the frame rule of separation logic on the subalgebra of local
monotone predicate transformers.  We use these rules for automated
verification condition generation with our tool. Formalising Morgan's
specification statement~\cite{Mor98} on the quantale structures, we
obtain tools for program construction and refinement with a frame law
with minimal effort. The predicate transformer semantics for
separation logic, instead of the usual state transformer
ones~\cite{COY07}, fits well into the power series approach and
simplifies the development.

The formalisation of the algebraic hierarchy from resource monoids to
predicate transformer algebras benefits from Isabelle's excellent
libraries for functions and, in particular, its integration of
automated theorem provers and SMT-solvers via the Sledgehammer tool.
These are optimised for equational reasoning and make the entire
development mostly automatic. In addition, Isabelle's reconstruction
of proof outputs provided by the external tools makes our verification
tools correct by construction.

At the data flow level, we currently use Isabelle's extant
infrastructure for the store, the heap and pointer-based data
structures. An interface to the control flow algebra is provided by
the assignment and mutation laws of separation logic and their
refinement counterparts. Isabelle's data flow models are linked
formally with our abstract separation algebra by soundness
proofs. Algebraic facts are then picked up automatically by Isabelle
for reasoning in the concrete model. The verification examples in the
last section of this article show that, at this layer, proofs may
require some user interaction, but a Sledgehammer-style integration of
optimised efficient provers and solvers for the data flow is an avenue
of future work.

The entire technical development has been formalised in Isabelle; all
proofs have been formally verified.  For this reason we show only some
example proofs which demonstrate the simplicity of algebraic
reasoning. The complete executable Isabelle theories can be found
online\footnote{https://github.com/vborgesfer/sep-logic}.

%%%%%%%%%%%%%%%%%%%%%%%%%%%%%%%%%%%%%%%%%%%%%%%%%%%%%%%%%%%%%%%%%%%%%%%%%%%%%%

\section{Partial Monoids, Quantales and Power
  Series} \label{S:preliminaries}

This section presents the algebraic structures that underlie our
approach to separation logic. Further details on power series and
lifting constructions can be found in~\cite{DHS14}.

A \emph{partial monoid} is a structure $(M,\cdot,1,\bot)$ such that
$(M,\cdot,1)$ is a monoid and $x\cdot \bot =\bot = \bot \cdot x$ holds
for all $x \in M\cup\{1,\bot\}$.  We often do not mention $\bot$ in
definitions. A partial monoid $M$ is \emph{commutative} if $x\cdot
y=y\cdot x$ for all $x,y\in M$. Henceforth $\cdot$ is used for a
general and $\ast$ for a commutative multiplication.

A \emph{quantale} is a structure $(Q,\le,\cdot,1)$ such that $(Q,\le)$
is a complete lattice, $(Q,\cdot)$ is a monoid and the distributivity
axioms
\begin{equation*}
	(\sum_{i\in I} x_i)\cdot y  = \sum_{i\in I} x_i\cdot y,\qquad\qquad 
   x\cdot (\sum_{i\in I}y_i)  =  \sum_{i\in I} (x\cdot y_i)
\end{equation*}
hold, where $\sum X$ denotes the supremum of a set $X\subseteq
Q$. Similarly, we write $\prod X$ for the infimum of $X$, and $0$ for
the least and $U$ for the greatest element of the lattice.  The
monotonicity laws $ x \le y \Rightarrow z\cdot x \le z\cdot y$ and
$x\le y \Rightarrow x\cdot z \le y\cdot z$ follow from distributivity.
A quantale is \emph{commutative} and \emph{partial} if the underlying
monoid is. It is \emph{distributive} if $x\sqcap (\sum_{i\in I} y_i)
=\sum_{i\in I} (x\sqcap y_i)$ and $x+ (\prod_{i\in I} y_i) =
\prod_{i\in I} (x+ y_i)$. In that case, the annihilation laws $x \cdot
0 = 0 =0 \cdot x$ follow from $\sum_{i\in\emptyset} x_i = \sum
\emptyset = 0$.  A \emph{boolean quantale} is a complemented
distributive quantale. The boolean quantale $\mathbb{B}$ of the
booleans, where multiplication coincides with join, is an important
example.

A \emph{power series} is a function $f:M\to Q$, from a
partial monoid $M$ into a quantale $Q$.  For $f,g:M\to Q$ and a family of functions
$f_i:M\to Q$, $i\in I$ we define
 \begin{equation*}
(f\cdot g)\ x =  \sum_{x=y\cdot z} f\ y\cdot g\ z,\qquad (\sum_{i\in I} f_i)\ x  = \sum_{i\in I} f_i\ x.
 \end{equation*}
 The composition $f\cdot g$ is called \emph{convolution}; the
 multiplication symbol is overloaded to be used on $M$, $Q$ and the
 function space $Q^M$.  The idea behind convolution is simple: element
 $x$ is split into $y$ and $z$, $f\ y$ and $g\ z$ are calculated in
 parallel, and their results are composed to form a value for the
 summation with respect to all possible splits of $x$. Because $x$
 ranges over $M$, the constant $\bot\notin M$ is excluded as a
 value. Hence undefined splittings of $x$ do not contribute to
 convolutions. In addition, $(f+g)\ x = f\ x + g\ x$ arises as a
 special case of the supremum.  Finally, we define the power series
 $\zero:M\to Q$ as $\zero = \lambda x.\ 0$ and $\unit:M\to Q$ as
 $\unit = \lambda x.\ \textbf{if}\ (x = 1)\ \textbf{then}\ 1\
 \textbf{else}\ 0$.

 The quantale structure lifts from $Q$ to the function space $Q^M$ of
 power series.
\begin{theorem}[\cite{DHS14}]\label{thm:quantale-lifting} 
  Let $M$ be a partial monoid. If $Q$ is a boolean quantale, then so
  is $(Q^M,\le,\cdot,\unit)$. If $M$ and $Q$ are commutative, then so is $Q^M$.
\end{theorem}

The power series approach generalises from one to $n$
dimensions~\cite{DHS14}. For separation logic, the
two-dimensional case with power series $f:S \times M\to Q$ from 
set $S$ and partial commutative monoid $M$ into the commutative
quantale $Q$ is needed.  Now
\begin{equation*}
(f \ast g)\ (x,y) = \sum_{y=y_1\ast y_2} f\ (x,y_1)\ast g\ (x,y_2),\qquad
  (\sum_{i\in I} f_i)\ (x,y)  = \sum_{i\in I} f_i\ (x,y).
\end{equation*}
The convolution $f\ast g$ acts solely on the second
coordinate. Finally, we define two-dimensional units as $\zero =
\lambda x,y.\ 0$ and $\unit = \lambda x,y.\ \textbf{if}\ (y = 1)\
\textbf{then}\ 1\ \textbf{else}\ 0$.
\begin{theorem}[\cite{DHS14}]\label{thm:sh-lifting}
  Let $S$ be a set and $M$ a partial commutative monoid. If $Q$ is a
  commutative boolean quantale, then so is $Q^{S\times M}$.
\end{theorem}
We have implemented partial monoids and quantale by using Isabelle's
type class and local infrastructure for engineering mathematical
structures, building on existing libraries for monoids, quantales and
complete lattices.

%%%%%%%%%%%%%%%%%%%%%%%%%%%%%%%%%%%%%%%%%%%%%%%%%%%%%%%%%%%%%%%%%%%%%%%

\section{Assertion Quantales}\label{S:assertion_quantale}

In language theory, power series have been introduced for modelling
formal languages.  In that case, $S$ is the free monoid $X^\ast$ and
$Q$ can be taken as a semiring $(Q,+,\cdot,0,1)$, because there are
only finitely many ways of splitting words into prefix/suffix pairs in
convolutions.  In the particular case of the boolean semiring
$\mathbb{B}$, where $\cdot$ is conjunction, power series $f:X^\ast\to
\mathbb{B}$ are interpreted as characteristic functions or predicates
that indicate whether or not a word is in a set. Since, in this case,
sets are languages, convolution specialises to $(f\cdot g)\ x =
\sum_{x=yz} f\ x \sqcap g\ y$, hence, identifying predicates with
their extensions, to the language product $p\cdot q=\{yz\mid y\in
p\wedge z\in q\}$.

More generally, we consider power series $S\to \mathbb{B}$ from a
partial monoid $S$ into the boolean quantale $\mathbb{B}$ and set up
the connection with separation logic. There, one is interested in
modelling assertions or predicates over the memory heap. 
The heap can be
represented abstractly by a resource monoid~\cite{COY07}, which is simply a
partial commutative monoid. In analogy to the
language case, an assertion $p$ of separation logic is a
boolean-valued function from a resource monoid $M$, hence a power
series $p:M\to\mathbb{B}$. Then Theorem~\ref{thm:quantale-lifting}
applies.
\begin{corollary}\label{cor:heap-ass-quantale}
  The assertions $\mathbb{B}^M$ over resource monoid $M$ form a 
  commutative boolean quantale with convolution as separating 
  conjunction. 
\end{corollary}
The logical structure of the assertion quantale $\mathbb{B}^M$ is as
follows. The predicate $\zero$ is a contradiction whereas $\unit$
holds of the empty resource and is false otherwise. The operations
${\rm \Sigma}$ and ${\rm \Pi}$ correspond to existential and universal
quantification; their finite cases yield conjunctions and
disjunctions. The order $\le$ is implication. Convolution becomes
\begin{equation*}
  (p\ast q)\ x=  \sum_{x=y\ast z} p\ y\sqcap q\ z.
\end{equation*}
By $x=y \ast z$, resource $x$ is separated into resources $y$ and
$z$. By $p\ y\sqcap q\ z$, the value of predicate $p$ on $y$ is
conjoined with that of $q$ on $z$. Finally, the supremum is true if
one of the conjunctions holds for some splitting of $x$.

As for languages, one can again identify predicates with their
extensions. Then
\begin{equation*}
  p\ast q = \{ y\ast z \in M\mid y\in p\wedge z\in q\},
\end{equation*}
and separating conjunction becomes a language product over
resources. The analogy to language theory is even more striking
when considering multisets over a finite set $X$ as resources, which
form the free commutative monoids over $X$.

Applications of separation logic, however, require program states
which are store-heap pairs.  Now Theorem~\ref{thm:sh-lifting} applies.
\begin{corollary}\label{cor:heap-ass-quantale}
  The assertions $\mathbb{B}^{S\times M}$ over set $S$ (the store) and
  resource monoid $M$ form a commutative boolean quantale with
  convolution as separating conjunction. For all $p,q:S\times M\to
  \mathbb{B}$, $s\in S$ and $h\in M$,
\begin{equation*}
    (p \ast q)\ (s,h) = \sum_{h=h_1\ast h_2} p\ (s,h_1)\sqcap\ q\ (s,h_2). 
\end{equation*}
\end{corollary}

Quantales carry a rich algebraic structure. Their distributivity laws
give rise to continuity or co-continuity properties.  Therefore, many
functions constructed from the quantale operations have adjoints as
well as fixpoints, which can be iterated to the first limit
ordinal. This is well known in denotational semantics and important
for our approach to program verification.  In particular, separating
conjunction $\ast$ distributes over arbitrary suprema in
$\mathbb{B}^M$ and $\mathbb{B}^{S\times M}$ and therefore has an upper
adjoint: the \emph{magic wand} operation $\wand$, which is widely used
in separation logic. In the quantale setting, the adjunction gives us
theorems for the magic wand for free.

One can think of the power series approach to separation logic as a
simple account of the category-theoretical approach to O'Hearn and
Pym's logic of bunched implication~\cite{OHearnP99} in which
convolution generalises to coends and the quantale lifting is embodied
by Day's construction~\cite{Day}. For the design of verification tools
and our implementation in Isabelle, this simplicity is certainly an
advantage.

%%%%%%%%%%%%%%%%%%%%%%%%%%%%%%%%%%%%%%%%%%%%%%%%%%%%%%%%%%%%%%%%%%%%%%%%%%%%%%

\section{Predicate Transformer Quantales} \label{S:pt_quantale}

Our algebraic approach to separation logic is based on predicate
transformers (cf. ~\cite{BvW99-book}).  This is in contrast to
previous state-transformer-based approaches and
implementations~\cite{COY07,NICTA,Tuerk}.  First of all, predicate
transformers are more amenable to algebraic
reasoning~\cite{BvW99-book}---simply because their source and target
types are similar. Second, the approach is more coherent within our
framework.  Predicate transformers can be seen once more as power
series and instances of Theorem~\ref{thm:quantale-lifting} describe
their algebras.

A \emph{state transformer} $f_R:A\to 2^B$ is often associated with a
relation $R\subseteq A\times B$ from set $A$ to set $B$ by defining $
f_R\ a=\{b \mid (a,b)\in R\}$. It can be lifted to a function $\langle
R\rangle:2^A\to 2^B$ defined by $\langle R\rangle X = \bigcup_{a\in X}
f_R\ a$ for all $X\subseteq A$. More importantly, state transformers
are lifted to \emph{predicate transformers} $\predT{R}:2^B\to 2^A$ by
defining
\begin{equation*}
  \predT{R}\ Y=\{x \mid f_R\ x\subseteq Y\}
\end{equation*}
for all $Y\subseteq B$. The modal box and diamond notation is
justified by the correspondance between diamond operators and Hoare
triples as well as box operators and weakest liberal precondition
operators in the context of modal semirings~\cite{MoellerStruth}. In
fact we obtain the adjunction $\langle R\rangle X \subseteq
Y\Leftrightarrow X\subseteq [R]Y$ from the above definitions.

Predicate transformers in $(2^A)^{2^B}$ form complete distributive
lattices~\cite{BvW99-book}. In the power series setting, this follows
from Theorem~\ref{thm:quantale-lifting} in two steps, ignoring the
monoidal structure. Since $\mathbb{B}$ forms a complete distributive
lattice, so do $2^B\cong \mathbb{B}^B$ and $2^A\cong\mathbb{B}^A$ in
the first step, and so does $(2^A)^{2^B}$ in the second one.

In addition, predicate transformers in $(2^A)^{2^A}$ form a monoid
under function composition and the identity function as the unit.  It
follows that those predicate transformers form a distributive
near-quantale, whereas the monotone predicate transformers in
$(2^A)^{2^A}$, which satisfy $p\le q\Rightarrow f\ p\le f\ q$, form a
distributive pre-quantale~\cite{BvW99-book}.

Here, \emph{near-quantale} means a quantale where $x\cdot (\sum_{i\in
  I}y_i) = \sum_{i\in I} (x\cdot y_i)$, the left distributivity law,
need not hold. We call \emph{pre-quantale} a near-quantale in which
the left monotonicity law $x\le y\Rightarrow z\cdot x\le z\cdot y$
holds. In these cases, the monoidal parts of the lifting are not
obtained with the power series technique because function composition
is not a convolution.

Adapting these results to separation logic requires the consideration
of assertion quantales $\mathbb{B}^M$ or $\mathbb{B}^{S\times M}$ with
store $S$ and resource monoid $M$ instead of the powerset algebra over
a set $A$. But these quantales can be lifted as boolean
algebras---disregarding a convolution on predicate transformers, which
is not needed for separation logic---and combined with the monoidal
structure of function composition as previously. This yields the
following result.
\begin{theorem}\label{thm:sep-pt-quantales}
  Let $S$ be a set, $M$ a resource monoid and $\mathbb{B}^{S\times M}$
  an assertion quantale. The monotone predicate transformers over
  $\mathbb{B}^{S\times M}$ form a distributive pre-quantale.
\end{theorem}
The proof requires showing that the predicate transformers over
$\mathbb{B}^{S\times M}$ form a near quantale and checking that the
monotone predicate transformers form a subalgebra of this
near-quantale. In fact, the unit predicate transformer is
monotone---which is the case---and that the quantale operations of
suprema, infima and composition preserve monotonicity. This is implied
by properties such as
\begin{equation*}
[R\cup S] = [R] \sqcap [S],\qquad [R;S]= [R]\cdot [S] = \lambda x.\ [R]\ ([S]\ x).
\end{equation*}

Monotone predicate transformers are powerful enough to derive the
standard inference rules of Hoare logic as verification conditions
(Section~\ref{sec:verification}) and the usual rules of Morgan's
refinement calculus (Section~\ref{sec:refinement}). Derivation of the
frame rule of separation logic, however, requires a smaller class of
predicate transformers defined as follows.

A state transformer $f$ is \emph{local} ~\cite{COY07} if $f(x \ast y)
\le (f\ x) \ast \{y\}$ for $x \ast y\neq \bot$. Intuitively, this
means that the effect of such a state transformer is restricted to a
part of the heap; see~\cite{COY07} for further discussion.
Analogously, we call a predicate transformer $F$ \emph{local} if
\begin{equation*}
(F\ p) \ast q \le F\ (p\ast q). p
\end{equation*}
It is easy to show that the two definitions are compatible.
\begin{lemma}\label{lem:local-prop2}
  State transformer $f_R$ is local iff predicate transformer
  $[R]$ is local.
\end{lemma}

The final theorem in this section establishes the local monotone
predicate transformers as a suitable algebraic framework for
separation logic.
\begin{theorem}\label{prop:pt-loc-quantale}
  Let $S$ be a set and $M$ a resource monoid. The local monotone
  predicate transformers over the assertion quantale
  $\mathbb{B}^{S\times M}$ form a distributive pre-quantale.
\end{theorem}
Once again it must be checked that the zero predicate transformer is
local---which is the case---and that the quantale operations preserve
locality and monotonicity.

We have implemented the whole approach in Isabelle; all theorems have
been formally verified, using mainly
Theorem~\ref{thm:quantale-lifting} for lifting to predicate
transformers. An alternative Isabelle implementation of predicate
transformers as lattices and Boolean algebras with operators is due to
Preoteasa~\cite{Preoteasa11}.

%%%%%%%%%%%%%%%%%%%%%%%%%%%%%%%%%%%%%%%%%%%%%%%%%%%%%%%%%%%%%%%%%%%%%%%%%%%%%%

\section{Verification Conditions} \label{sec:verification}

The pre-quantale of local monotone predicate transformers supports the
derivation of verification conditions by equational reasoning. A
standard set of such conditions are the inference rules of Hoare
logic. For sequential programs, Hoare logic provides one inference
rule per program construct. This suffices to eliminate the control
structure of a program and generate verification conditions for the
data flow.

The quantale setting also guarantees that the finite iteration
$F^\ast$ of a predicate transformer is well defined. This supports a
shallow algebraic embedding of a simple while language with the
usual pseudocode for the verification of imperative programs.

First we lift predicates to predicate transformers~\cite{BvW99-book}:
\begin{equation*}
  \predT{p} = \lambda q.\ \overline{p} + q.
\end{equation*}
With predicates modelled as relational subidentities, this definition
is justified by the lifting from the previous section: $(s, s ) \in
\predT{p}\ q\ $ iff $(s, s) \in p \Rightarrow (s, s) \in q$.

Second, we change notation, to use descriptive while program syntax
for predicate transformers. We write $\sskip$ for the quantale unit
(the identity function) and $;$ for function composition.  We encode
the semantics of conditionals and while loops as
\begin{equation*}
	\ifstat{p}{F}{G} = \predT{p}\cdot F + \predT{\overline{p}}\cdot G, \qquad 
	\whileloop{p}{F} = (\predT{p}\cdot F)^\ast\cdot\predT{\overline{p}}. 
\end{equation*}
Third, we provide the standard assertions notation for programs via
Hoare triple syntax:
\begin{equation*}
	\triple{p}{F}{q} \Leftrightarrow p \le F\ q.
\end{equation*}
Box notation shows that $\triple{p}{[R]}{q}\Leftrightarrow p\le [R]q$
for relational program $R$. Thus $[R]q=\mathit{wlp}(R,q)$ is the
standard weakest liberal precondition of program $R$ and postcondition
$q$. It also explains our slight abuse of relational or imperative
notation for predicate transformers: e.g. we write $[R];[S]$ instead
of $[R]\cdot [S]$ because the latter expression is equal to $[R;S]$,
as indicated in the previous section.

\begin{proposition}
  Let $p,q,r,p',q'\in \mathbb{B}^{S\times M}$ be predicates. Let
  $F,G,H$ be monotone predicate transformers over $\mathbb{B}^{S\times
    M}$, with $H$ being local.  Then the rules of propositional
  Hoare logic (no assignment rule) and the frame rule of
  separation logic are derivable.
\begin{align*}
	& \triple{p}{\sskip}{p}, \\
	p \le p' \wedge  q' \le q\ \wedge \triple{p'}{F}{q'} \ \Rightarrow\ &
		\triple{p}{F}{q}, \\
	\triple{p}{F}{r}\ \wedge \triple{r}{G}{q}\ \Rightarrow\ &
		\triple{p}{F; G}{q}, \\
	\triple{p \sqcap r}{F}{G}\ \wedge \triple{p \sqcap \overline{r}}{G}{q}\ \Rightarrow\ &
		\triple{p}{\ifstat{r}{F}{G}}{q}, \\
	\triple{p \sqcap q}{F}{p}\ \Rightarrow\ & \triple{p}{\whileloop{q}{F}}{\overline{b} \sqcap p}, \\
	\triple{p}{H}{q} \Rightarrow\ & \triple{p \ast r}{H}{q \ast r}.
\end{align*}
\end{proposition}
\begin{proof} 
  We derive the frame rule as an example. Suppose $p \le H\ q$. Then,
  by isotonicity of $\ast$ and locality,
$p\ast r\le (H\ q) \ast r \le H (q \ast r)$.\qed
\end{proof}
The remaining derivations are equally simple and fully automatic in
Isabelle.

%%%%%%%%%%%%%%%%%%%%%%%%%%%%%%%%%%%%%%%%%%%%%%%%%%%%%%%%%%

\section{Refinement Laws}\label{sec:refinement}

To demonstrate the power of the predicate transformer approach to
separation logic we now outline its applicability to local reasoning
in program construction and transformation. More precisely, we show
that the standard laws of Morgan's refinement calculus~\cite{Mor98}
plus an additional framing law for resources can be derived with
little effort.  It only requires defining one single additional
concept---Morgan's specification statement---which already exists in
every predicate transformer quantale.

 Formally, for predicates $p,q\in\mathbb{B}^{S\times M}$, we define
 the \emph{specification statement} as
\begin{equation*}
  \spec{p}{q} = \sum \{F \mid p\le F\ q\}.
\end{equation*}
It models the most general predicate transformer or program which
links postcondition $q$ with precondition $p$. It is easy to see that
$\triple{p}{F}{q}\Leftrightarrow p \le F\ q$, which entails the
characteristic properties $\triple{p}{\spec{p}{q}}{q}$ and
$\triple{p}{F}{q} \Rightarrow F \le \spec{p}{q}$ of the specification
statement: program $\spec{p}{q}$ relates precondition $p$ with
postcondition $q$ whenever it terminates; and it is the largest
program with that property. It is easy to check that specification
statements over the pre-quantale of local monotone predicate
transformers are themselves local and monotone.

Like Hoare logic, Morgan's basic refinement calculus provides one
refinement law per program construct. Once more we ignore assignments
at this stage. We also switch to standard refinement notation with
refinement order $\sqsubseteq$ being the converse of $\le$. 

\begin{proposition}\label{P:morganderive}
  For $p,q,r,p',q'\in\mathbb{B}^{S\times M}$, and predicate
  transformer $F$ the following refinement laws are derivable in the
  algebra of local monotone predicate transformers.
\begin{align*}
	p \le q \Rightarrow \spec{p}{q} &\sqsubseteq \sskip, \\
	p' \le p \wedge q \le q' \Rightarrow \spec{p}{q} &\sqsubseteq \spec{p'}{q'}, \\
			\spec{0}{1} &\sqsubseteq F, \\
	F &\sqsubseteq \spec{1}{0}, \\
	\spec{p}{q} &\sqsubseteq \spec{p}{r}; \spec{r}{q}, \\
	\spec{p}{q} &\sqsubseteq \ifstat{b}{\spec{b\sqcap p}{q}}{\spec{\overline{b}\sqcap p}{q}}, \\
	\spec{p}{\overline{b}\sqcap p} &\sqsubseteq \whileloop{b}{\spec{b\sqcap p}{p}}, \\
	\spec{p \ast r}{q \ast r} &\sqsubseteq \spec{p}{q}. 
\end{align*}
\end{proposition}
\begin{proof} 
Using the frame rule, we derive the framing law as an example:
\begin{equation*}
	\triple{p}{\spec{p}{q}}{q} \Rightarrow \triple{p \ast r}{\spec{p}{q}}{q \ast r} 
		\Leftrightarrow \spec{p \ast r}{q \ast r} \sqsubseteq \spec{p}{q}
\end{equation*} 
\qed
\end{proof}
The proofs of the other refinement laws are equally simple, using the
corresponding Hoare rules in their proofs.  They are fully automatic
in Isabelle.

%%%%%%%%%%%%%%%%%%%%%%%%%%%%%%%%%%%%%%%%%%%%%%%%%%%%%%%%%%

\section{Principles of Tool Design}\label{S:tool_design}

The previous sections have introduced a new algebraic approach to
separation logic based on power series and quantales with separating
conjunction modelled algebraically as convolution, that is, a language
product over resources. The theory hierarchy from partial monoids to
predicate transformer quantales has been formalised in Isabelle/HOL,
much of which was highly automatic and required only a moderate
effort. It benefits, to a large extent, from Isabelle's integrated
first-order theorem proving, SMT-solving and counterexample generation
technology.  These tools are highly optimised for equational
reasoning, interacting efficiently with the algebraic layer.

\begin{figure}[h]
  \centering 
\setlength{\tabcolsep}{10pt}
\renewcommand{\arraystretch}{2.2}
        \begin{tabular}{c|c|c}
algebra &  intermediate semantics   &  concrete semantics\\
\hline 
 control flow & abstract data flow & concrete data flow\\
 verification conditions  & - &
 \textbf{verification tools}\\
transformation laws & - &\textbf{construction tools}
        \end{tabular}
  \caption{Principles of Tool Design}
  \label{fig:design-principles}
\end{figure}

Another important feature is that the mathematical structures
formalised in Isabelle are all polymorphic---their elements can have
various types. Isabelle's type classes and locales, which have been
used for implementing mathematical hierarchies, then allow us to link
these abstract algebras with various concrete models, that is,
quantales with predicate transformers, predicate transformers with
binary relations and functions which update program states. In
particular, abstract resource monoids are linked with various concrete
models for resources, including the heap. By formalising these
soundness results in Isabelle, theorems are automatically propagated
across classes and models. Those proved for quantales, for instance,
become available automatically for predicate transformers over
concrete detailed store-heap models.

Finally, our development can build on excellent Isabelle libraries and
decades-long experience in reasoning with functions and relations, all
sorts of data structures and data types.  In particular, for program
construction and verification with separation logic, Isabelle provides
support for reasoning with pointers and the
heap~\cite{MehtaN05,Weber04}.

These features suggest a principled approach to program verification
and construction in Isabelle in which the control flow layer is
cleanly separated from the data flow layer.  The control flow is
modelled in a lightweight way within suitable algebras, which makes
tool design fast, simple and automatic, including the development of
verification conditions, transformation and refinement laws. Their
application can then be automated by programming Isabelle tactics.
The data flow can, by and large, rely on existing Isabelle libraries,
which have previously been designed for verification purposes. The
interface between these layers is provided, at an abstract level, by
soundness theorems and, at the concrete level, by assignment laws and
similar laws that link data and control. The approach has previously
been applied to simple while programs~\cite{sefm2014} and the
rely-guarantee approach~\cite{fm2014}.  Its
basic features a summarised in Figure~\ref{fig:design-principles}.

For separation logic, the algebraic structures used are partial
monoids, quantales and power series. The intermediate semantics is
provided by predicate transformers over assertion quantales based on
stores and resource monoids. The concrete models yield detailed
descriptions of the store and heap. These are explained in the
remaining sections of this article.

% In fact, we could have lifted the entire power series approach to the
% quantale level by lifting function application $F\ p$ to function
% composition $F\cdot [p]$. Pragmatically, however, due to Isabelle's
% excellent function libraries, this might be a disadvantage.

%%%%%%%%%%%%%%%%%%%%%%%%%%%%%%%%%%%%%%%%%%%%%%%%%%%%%%%%%%%%%%%%%%%%%%%%%%%%%%

\section{Integration of Data Flow}\label{S:data_flow}

This section describes the integration of the data flow layer into our
Isabelle tools for program verification and construction.

As previously mentioned, program states are store-heap pairs $(s, h)$.
Program stores are implemented in Isabelle as records of program
variables, each of which has a \emph{retrieve} and an \emph{update}
function.  This approach is polymorphic and supports variables of any
Isabelle type.  For instance, Isabelle's built-in list data type and
list libraries can be used to reason about list-based programs.  Heaps
have been modelled in Isabelle as partial functions on
$\mathbb{N}$~\cite{MehtaN05,Weber04}; they therefore have type
$\mathbf{nat} \to \mathbf{nat\ option}$.

We implement assignments first as functions from states to states,
\begin{equation*}
	(`x := e) = \lambda (s, h).\ (x\_update\ s\ e, h),
\end{equation*}
where $`x$ is a program variable, $x\_update$ the update function for
$`x$, $(s, h)$  a state and $e$ an evaluated expression of the same type as $`x$.

Separation logic also requires a notion of mutation or heap update.
In Isabelle, it is first implemented in similar fashion as
\begin{equation*}
	(@e := e') = \lambda (s, h).\ (s, h[e \mapsto e']),
\end{equation*}
where $e$ and $e'$ are expressions that evaluate to natural numbers,
and $h[e \mapsto e']$ is the function that maps $e$ to $e'$ and is the
same as $h$ for the remaining expressions.

Secondly, we lift assignment and mutation functions to
predicate transformers as
\begin{equation*}
	\predT{f} = \lambda q.\ q \cdot f,
\end{equation*}
where $\cdot$ denotes function composition, as usual. This is
consistent with the definition of lifting in Section
\ref{S:pt_quantale}. As previously, we generally do not write the
lifting brackets explicitly, identifying program pseudocode with
predicate transformers to simplify verification notation.

With this infrastructure in place we can prove Hoare's assignment rule
and Reynolds' mutation rule for separation logic in the concrete heap
model.
\begin{proposition}
  The following rules are derivable in the concrete store-heap model:
\begin{equation*}
	p \le q[e/`x] \Rightarrow\ \triple{p}{(`x := e)}{q}, \qquad
	\triple{(e \mapsto -) * r}{(@e := e')}{(e \mapsto e') * r},
\end{equation*}
writing $q[e/`x]$ for the substitution of variable $`x$ by expression
$e$ in $q$ as well as $e \mapsto e'$ and $e\mapsto -$ for the
singleton heaps mapping $e$ to $e'$ and to any value, respectively.
\end{proposition}

The resulting set of inference rules for separation logic allows us to
implement the Isabelle proof tactic \emph{hoare}, which generates all
verification conditions automatically and eliminate the entire control
structure when the invariants of while loops are annotated.

One can use the assignment rule to derive its refinement counterparts:
\begin{align*}
	p \le q[e/`x] &\Rightarrow \spec{p}{q} \sqsubseteq
		(`x := e), \\
	q' \le q[e/`x] &\Rightarrow \spec{p}{q} \sqsubseteq
		\spec{p}{q'};\ (`x := e), \\
	p' \le p[e/`x] &\Rightarrow \spec{p}{q} \sqsubseteq
		(`x := e);\ \spec{p'}{q}.
\end{align*}
The second and third laws are called the \emph{following} and
\emph{leading} refinement law for assignments~\cite{Mor98}. They are
useful for program construction. We have derived analogous laws for
mutation. We have also implemented the tactic \emph{refinement}, which
automatically tries to apply all the rules of this refinement calculus
in construction steps of pointer programs.

%%%%%%%%%%%%%%%%%%%%%%%%%%%%%%%%%%%%%%%%%%%%%%%%%%%%%%%%%%%%

\section{Examples}\label{S:examples}

To show our approach at work, we present the obligatory correctness
proof of the classical in situ linked-list reversal algorithm. The
post-hoc verification of this algorithm in Isabelle has been
considered before~\cite{MehtaN05, Weber04}.  However, we follow
Reynolds~\cite{Reynolds02}, who gave an informal annotated proof, and
reconstruct his proof step-by-step in refinement style. As usual for
verification with interactive theorem provers, functional
specifications are related to imperative data structures. The former
are defined recursively and hence amenable to proof by induction.

First, we define two inductive predicates on the heap.  The first one
creates a contiguous heap from a position $e$ using Isabelle's
functional lists as its representation. That is, by induction on the
structure of the list,
\begin{alignat*}{3}
	e &\mapslist \emptylist \ &=& \ \eemp, \\
	e &\mapslist (t\#ts) \ &=& \ (e\ \mapsto\ t) \ast (e + 1 \mapslist ts),
\end{alignat*}
where $\emptylist$ denotes the empty list, $t\#ts$ denotes
concatenation of element $t$ with list $ts$, $e\ \mapsto\ t$ is again
a singleton heap predicate and $\eemp$ states that the heap is empty.

The second predicate indicates whether or not a heap, starting from position $i$,
contains the linked list represented as a functional list:
\begin{alignat*}{3}
	list\ i\ &\emptylist \ & =&\  (i = 0) \wedge \eemp, \\
	list\ i\ &(j\#js) \ &=& \  (i \neq 0) \wedge (\exists k.\ i \mapslist [j,k]\ \ast\ list\ k\ js).
\end{alignat*}
This is  Reynolds' definition, it uses separating conjunction
instead of plain conjunction.

We can now reconstruct Reynolds' proof relative to the standard
recursive function $\mathit{rev}$ for functional list reversal. The
initial specification statement is
\begin{equation*}
\spec{list\ {\isacharbackquote}i\ A_0}
	{list\ {\isacharbackquote}j\ {\isacharparenleft}rev\ A_0{\isacharparenright}},
\end{equation*}
where $A_0$ is the input list and $`i$ and $`j$ are pointers to the head of the list
on the heap.

\begin{figure}
\begin{isabellebody}
\isanewline

{\isasymlbrakk}
	list\ {\isacharbackquote}i\ A$_0${\isacharcomma}\ 
	list\ {\isacharbackquote}j\ {\isacharparenleft}rev\ A$_0${\isacharparenright}
{\isasymrbrakk}
% First refinement
\isanewline\ \ \ {\isasymsqsubseteq}\hfill $(1)$\isanewline
{\isacharbackquote}j\ {\isacharcolon}{\isacharequal}\ {\isadigit{0}}{\isacharsemicolon}\ 
\isanewline
{\isasymlbrakk}
	$\exists$A\ B{\isachardot}\ {\isacharparenleft}list\ {\isacharbackquote}i\ A\ {\isacharasterisk}\ list\ {\isacharbackquote}j\ B\ {\isacharparenright}\ $\wedge$\ {\isacharparenleft}rev\ A$_0${\isacharparenright}\ {\isacharequal}\ {\isacharparenleft}rev\ A{\isacharparenright}\ {\isacharat}\ B{\isacharcomma}\ 
list\ {\isacharbackquote}j\ {\isacharparenleft}rev\ A$_0${\isacharparenright}\ {\isasymrbrakk}\ 
% Second refinement
\isanewline\ \ \ {\isasymsqsubseteq}\hfill $(2)$\isanewline
\textcolor{blue}{{\isacharbackquote}j\ {\isacharcolon}{\isacharequal}\ {\isadigit{0}}{\isacharsemicolon}}\isanewline
\textcolor{blue}{$\wwhile$\ {\isacharbackquote}i\ {\isasymnoteq}\ {\isadigit{0}}\ $\ddo$}\isanewline
\ \ {\isasymlbrakk}{\isacharparenleft}$\exists$A\ B{\isachardot}\ {\isacharparenleft}list\ {\isacharbackquote}i\ A\ {\isacharasterisk}\ list\ {\isacharbackquote}j\ B{\isacharparenright}\ $\wedge$\ {\isacharparenleft}rev\ A$_0${\isacharparenright}\ {\isacharequal}\ {\isacharparenleft}rev\ A{\isacharparenright}\ {\isacharat}\ B{\isacharparenright}\ $\wedge$\ {\isacharbackquote}i\ {\isasymnoteq}\ {\isadigit{0}}{\isacharcomma}\ \isanewline
\ \ \ \ $\exists$A\ B{\isachardot}\ {\isacharparenleft}list\ {\isacharbackquote}i\ A\ {\isacharasterisk}\ list\ {\isacharbackquote}j\ B{\isacharparenright}\ $\wedge$\ {\isacharparenleft}rev\ A$_0${\isacharparenright}\ {\isacharequal}\ {\isacharparenleft}rev\ A{\isacharparenright}\ {\isacharat}\ B\ {\isasymrbrakk}\isanewline
\textcolor{blue}{$\ood$}\

\isanewline 

{\isasymlbrakk}{\isacharparenleft}$\exists$A\ B{\isachardot}\ {\isacharparenleft}list\ {\isacharbackquote}i\ A\ {\isacharasterisk}\ list\ {\isacharbackquote}j\ B{\isacharparenright}\ $\wedge$\ {\isacharparenleft}rev\ A$_0${\isacharparenright}\ {\isacharequal}\ {\isacharparenleft}rev\ A{\isacharparenright}\ {\isacharat}\ B{\isacharparenright}\ $\wedge$\ {\isacharbackquote}i\ {\isasymnoteq}\ {\isadigit{0}}{\isacharcomma}\ \isanewline
\ \ $\exists$A\ B{\isachardot}\ {\isacharparenleft}list\ {\isacharbackquote}i\ A\ {\isacharasterisk}\ list\ {\isacharbackquote}j\ B{\isacharparenright}\ $\wedge$\ {\isacharparenleft}rev\ A$_0${\isacharparenright}\ {\isacharequal}\ {\isacharparenleft}rev\ A{\isacharparenright}\ {\isacharat}\ B\ {\isasymrbrakk}
% Third refinement
\isanewline\ \ \ {\isasymsqsubseteq}\hfill $(3)$\isanewline
{\isasymlbrakk}{\isacharparenleft}$\exists$a\ A\ B\ k{\isachardot}\ {\isacharparenleft}{\isacharbackquote}i\ {\isacharbrackleft}{\isasymmapsto}{\isacharbrackright}\ {\isacharbrackleft}a{\isacharcomma}k{\isacharbrackright}\ {\isacharasterisk}\ list\ k\ A\ {\isacharasterisk}\ list\ {\isacharbackquote}j\ B{\isacharparenright}\ $\wedge$\ {\isacharparenleft}rev\ A$_0${\isacharparenright}\ {\isacharequal}\ {\isacharparenleft}rev\ {\isacharparenleft}a{\isacharhash}A{\isacharparenright}{\isacharparenright}\ {\isacharat}\ B{\isacharparenright}\ $\wedge$\ {\isacharbackquote}i\ {\isasymnoteq}\ {\isadigit{0}}{\isacharcomma}\ \isanewline
\ \ $\exists$A\ B{\isachardot}\ {\isacharparenleft}list\ {\isacharbackquote}i\ A\ {\isacharasterisk}\ list\ {\isacharbackquote}j\ B{\isacharparenright}\ $\wedge$\ {\isacharparenleft}rev\ A$_0${\isacharparenright}\ {\isacharequal}\ {\isacharparenleft}rev\ A{\isacharparenright}\ {\isacharat}\ B\ {\isasymrbrakk}
% Forth refinement
\isanewline\ \ \ {\isasymsqsubseteq}\hfill $(4)$\isanewline
{\isacharbackquote}k\ {\isacharcolon}{\isacharequal}\ {\isacharat}{\isacharparenleft}{\isacharbackquote}i\ {\isacharplus}\ {\isadigit{1}}{\isacharparenright}{\isacharsemicolon}\isanewline
{\isasymlbrakk}{\isacharparenleft}$\exists$a\ A\ B{\isachardot}\ {\isacharparenleft}{\isacharbackquote}i\ {\isacharbrackleft}{\isasymmapsto}{\isacharbrackright}\ {\isacharbrackleft}a{\isacharcomma}{\isacharbackquote}k{\isacharbrackright}\ {\isacharasterisk}\ list\ {\isacharbackquote}k\ A\ {\isacharasterisk}\ list\ {\isacharbackquote}j\ B{\isacharparenright}\ $\wedge$\ {\isacharparenleft}rev\ A$_0${\isacharparenright}\ {\isacharequal}\ {\isacharparenleft}rev\ {\isacharparenleft}a{\isacharhash}A{\isacharparenright}{\isacharparenright}\ {\isacharat}\ B{\isacharparenright}\ $\wedge$\ {\isacharbackquote}i\ {\isasymnoteq}\ {\isadigit{0}}{\isacharcomma}\ \isanewline
\ \ $\exists$A\ B{\isachardot}\ {\isacharparenleft}list\ {\isacharbackquote}i\ A\ {\isacharasterisk}\ list\ {\isacharbackquote}j\ B{\isacharparenright}\ $\wedge$\ {\isacharparenleft}rev\ A$_0${\isacharparenright}\ {\isacharequal}\ {\isacharparenleft}rev\ A{\isacharparenright}\ {\isacharat}\ B\ {\isasymrbrakk}
% Fifth refinement
\isanewline\ \ \ {\isasymsqsubseteq}\hfill $(5)$\isanewline
{\isacharbackquote}k\ {\isacharcolon}{\isacharequal}\ {\isacharat}{\isacharparenleft}{\isacharbackquote}i\ {\isacharplus}\ {\isadigit{1}}{\isacharparenright}{\isacharsemicolon}\isanewline
{\isacharat}{\isacharparenleft}{\isacharbackquote}i\ {\isacharplus}\ {\isadigit{1}}{\isacharparenright}\ {\isacharcolon}{\isacharequal}\ {\isacharbackquote}j{\isacharsemicolon}\isanewline
{\isasymlbrakk}{\isacharparenleft}$\exists$a\ A\ B{\isachardot}\ {\isacharparenleft}{\isacharbackquote}i\ {\isacharbrackleft}{\isasymmapsto}{\isacharbrackright}\ {\isacharbrackleft}a{\isacharcomma}{\isacharbackquote}j{\isacharbrackright}\ {\isacharasterisk}\ list\ {\isacharbackquote}k\ A\ {\isacharasterisk}\ list\ {\isacharbackquote}j\ B{\isacharparenright}\ $\wedge$\ {\isacharparenleft}rev\ A$_0${\isacharparenright}\ {\isacharequal}\ {\isacharparenleft}rev\ {\isacharparenleft}a{\isacharhash}A{\isacharparenright}{\isacharparenright}\ {\isacharat}\ B{\isacharparenright}\ $\wedge$\ {\isacharbackquote}i\ {\isasymnoteq}\ {\isadigit{0}}{\isacharcomma}\ \isanewline
\ \ $\exists$A\ B{\isachardot}\ {\isacharparenleft}list\ {\isacharbackquote}i\ A\ {\isacharasterisk}\ list\ {\isacharbackquote}j\ B{\isacharparenright}\ $\wedge$\ {\isacharparenleft}rev\ A$_0${\isacharparenright}\ {\isacharequal}\ {\isacharparenleft}rev\ A{\isacharparenright}\ {\isacharat}\ B\ {\isasymrbrakk}
% Sixth refinement
\isanewline\ \ \ {\isasymsqsubseteq}\hfill $(6)$\isanewline
{\isacharbackquote}k\ {\isacharcolon}{\isacharequal}\ {\isacharat}{\isacharparenleft}{\isacharbackquote}i\ {\isacharplus}\ {\isadigit{1}}{\isacharparenright}{\isacharsemicolon}\isanewline
{\isacharat}{\isacharparenleft}{\isacharbackquote}i\ {\isacharplus}\ {\isadigit{1}}{\isacharparenright}\ {\isacharcolon}{\isacharequal}\ {\isacharbackquote}j{\isacharsemicolon}\isanewline
{\isasymlbrakk}{\isacharparenleft}$\exists$a\ A\ B{\isachardot}\ {\isacharparenleft}list\ {\isacharbackquote}k\ A\ {\isacharasterisk}\ {\isacharbackquote}i\ {\isacharbrackleft}{\isasymmapsto}{\isacharbrackright}\ {\isacharbrackleft}a{\isacharcomma}{\isacharbackquote}j{\isacharbrackright}\ {\isacharasterisk}\ list\ {\isacharbackquote}j\ B{\isacharparenright}\ $\wedge$\ {\isacharparenleft}rev\ A$_0${\isacharparenright}\ {\isacharequal}\ {\isacharparenleft}rev\ {\isacharparenleft}a{\isacharhash}A{\isacharparenright}{\isacharparenright}\ {\isacharat}\ B{\isacharparenright}\ $\wedge$\ {\isacharbackquote}i\ {\isasymnoteq}\ {\isadigit{0}}{\isacharcomma}\ \isanewline
\ \ $\exists$A\ B{\isachardot}\ {\isacharparenleft}list\ {\isacharbackquote}i\ A\ {\isacharasterisk}\ list\ {\isacharbackquote}j\ B{\isacharparenright}\ $\wedge$\ {\isacharparenleft}rev\ A$_0${\isacharparenright}\ {\isacharequal}\ {\isacharparenleft}rev\ A{\isacharparenright}\ {\isacharat}\ B\ {\isasymrbrakk}
% Seventh refinement
\isanewline\ \ \ {\isasymsqsubseteq}\hfill $(7)$\isanewline
{\isacharbackquote}k\ {\isacharcolon}{\isacharequal}\ {\isacharat}{\isacharparenleft}{\isacharbackquote}i\ {\isacharplus}\ {\isadigit{1}}{\isacharparenright}{\isacharsemicolon}\isanewline
{\isacharat}{\isacharparenleft}{\isacharbackquote}i\ {\isacharplus}\ {\isadigit{1}}{\isacharparenright}\ {\isacharcolon}{\isacharequal}\ {\isacharbackquote}j{\isacharsemicolon}\isanewline
{\isasymlbrakk}{\isacharparenleft}$\exists$a\ A\ B{\isachardot}\ {\isacharparenleft}list\ {\isacharbackquote}k\ A\ {\isacharasterisk}\ list\ {\isacharbackquote}i\ {\isacharparenleft}a{\isacharhash}B{\isacharparenright}\ {\isacharparenright}\ $\wedge$\ {\isacharparenleft}rev\ A$_0${\isacharparenright}\ {\isacharequal}\ {\isacharparenleft}rev\ A{\isacharparenright}\ {\isacharat}\ {\isacharparenleft}a{\isacharhash}B{\isacharparenright}\ {\isacharparenright}\ $\wedge$\ {\isacharbackquote}i\ {\isasymnoteq}\ {\isadigit{0}}{\isacharcomma}\ \isanewline
\ \ $\exists$A\ B{\isachardot}\ {\isacharparenleft}list\ {\isacharbackquote}i\ A\ {\isacharasterisk}\ list\ {\isacharbackquote}j\ B{\isacharparenright}\ $\wedge$\ {\isacharparenleft}rev\ A$_0${\isacharparenright}\ {\isacharequal}\ {\isacharparenleft}rev\ A{\isacharparenright}\ {\isacharat}\ B\ {\isasymrbrakk}
% Eighth refinement
\isanewline\ \ \ {\isasymsqsubseteq}\hfill $(8)$\isanewline
{\isacharbackquote}k\ {\isacharcolon}{\isacharequal}\ {\isacharat}{\isacharparenleft}{\isacharbackquote}i\ {\isacharplus}\ {\isadigit{1}}{\isacharparenright}{\isacharsemicolon}\isanewline
{\isacharat}{\isacharparenleft}{\isacharbackquote}i\ {\isacharplus}\ {\isadigit{1}}{\isacharparenright}\ {\isacharcolon}{\isacharequal}\ {\isacharbackquote}j{\isacharsemicolon}\isanewline
{\isasymlbrakk}{\isacharparenleft}$\exists$A\ B{\isachardot}\ {\isacharparenleft}list\ {\isacharbackquote}k\ A\ {\isacharasterisk}\ list\ {\isacharbackquote}i\ B{\isacharparenright}\ $\wedge$\ {\isacharparenleft}rev\ A$_0${\isacharparenright}\ {\isacharequal}\ {\isacharparenleft}rev\ A{\isacharparenright}\ {\isacharat}\ B{\isacharparenright}\ $\wedge$\ {\isacharbackquote}i\ {\isasymnoteq}\ {\isadigit{0}}{\isacharcomma}\ \isanewline
\ \ $\exists$A\ B{\isachardot}\ {\isacharparenleft}list\ {\isacharbackquote}i\ A\ {\isacharasterisk}\ list\ {\isacharbackquote}j\ B{\isacharparenright}\ $\wedge$\ {\isacharparenleft}rev\ A$_0${\isacharparenright}\ {\isacharequal}\ {\isacharparenleft}rev\ A{\isacharparenright}\ {\isacharat}\ B\ {\isasymrbrakk}
% Ninth refinement
\isanewline\ \ \ {\isasymsqsubseteq}\hfill $(9)$\isanewline
{\isacharbackquote}k\ {\isacharcolon}{\isacharequal}\ {\isacharat}{\isacharparenleft}{\isacharbackquote}i\ {\isacharplus}\ {\isadigit{1}}{\isacharparenright}{\isacharsemicolon}\isanewline
{\isacharat}{\isacharparenleft}{\isacharbackquote}i\ {\isacharplus}\ {\isadigit{1}}{\isacharparenright}\ {\isacharcolon}{\isacharequal}\ {\isacharbackquote}j{\isacharsemicolon}\isanewline
{\isacharbackquote}j\ {\isacharcolon}{\isacharequal}\ {\isacharbackquote}i{\isacharsemicolon}\isanewline
{\isasymlbrakk}{\isacharparenleft}$\exists$A\ B{\isachardot}\ {\isacharparenleft}list\ {\isacharbackquote}k\ A\ {\isacharasterisk}\ list\ {\isacharbackquote}j\ B{\isacharparenright}\ $\wedge$\ {\isacharparenleft}rev\ A$_0${\isacharparenright}\ {\isacharequal}\ {\isacharparenleft}rev\ A{\isacharparenright}\ {\isacharat}\ B{\isacharparenright}\ $\wedge$\ {\isacharbackquote}i\ {\isasymnoteq}\ {\isadigit{0}}{\isacharcomma}\ \isanewline
\ \ $\exists$A\ B{\isachardot}\ {\isacharparenleft}list\ {\isacharbackquote}i\ A\ {\isacharasterisk}\ list\ {\isacharbackquote}j\ B{\isacharparenright}\ $\wedge$\ {\isacharparenleft}rev\ A$_0${\isacharparenright}\ {\isacharequal}\ {\isacharparenleft}rev\ A{\isacharparenright}\ {\isacharat}\ B\ {\isasymrbrakk}
% Tenth and last refinement
\isanewline\ \ \ {\isasymsqsubseteq}\hfill $(10)$\isanewline
\textcolor{blue}{{\isacharbackquote}k\ {\isacharcolon}{\isacharequal}\ {\isacharat}{\isacharparenleft}{\isacharbackquote}i\ {\isacharplus}\ {\isadigit{1}}{\isacharparenright}{\isacharsemicolon}\isanewline
{\isacharat}{\isacharparenleft}{\isacharbackquote}i\ {\isacharplus}\ {\isadigit{1}}{\isacharparenright}\ {\isacharcolon}{\isacharequal}\ {\isacharbackquote}j{\isacharsemicolon}\isanewline
{\isacharbackquote}j\ {\isacharcolon}{\isacharequal}\ {\isacharbackquote}i{\isacharsemicolon}\isanewline
{\isacharbackquote}i\ {\isacharcolon}{\isacharequal}\ {\isacharbackquote}k}
\end{isabellebody}
\caption{\emph{In situ} list reversal by refinement. The first block
shows the refinement up to the introduction of the while-loop. The
second block shows the refinement of the body of that loop.}
\label{fig:ref-reversal}
\end{figure}

The main idea behind Reynolds' proof is to split the heap into two
lists, initially $A_0$ and an empty list, and then iteratively swing
the pointer of the first element of the first list to the second list.
The full proof is shown in Figure \ref{fig:ref-reversal}; we now
explain its details.

In $(1)$, we strengthen the precondition, splitting the heap into two
lists $A$ and $B$, and inserting a variable $`j$ initially assigned to
$0$ (or \emph{null}). The equation $(rev\ A_0) = (rev\ A)\ @\ B$ then
holds of these lists, where $@$ denotes the append operation on linked
lists. Justifying this step in Isabelle requires calling the
\emph{refinement} tactic from Section \ref{S:tool_design}, which
applies the leading law for assignment. This obliges us to prove that
the lists $A$ and $B$ \emph{de facto} exist, which is discharged
automatically by calling Isabelle's \emph{force} tactic.  In fact, $8$
out of the $10$ proof steps in our construction are essentially
automatic: they only require calling \emph{refinement} followed by
Isabelle's \emph{force} or \emph{auto} provers.

The new precondition generated then becomes the loop invariant of the
algorithm. It allows us to refine our specification statement to a
while loop in step $(2)$, where we iterate $`i$ until it becomes $0$.
Calling the \emph{refinement} tactic applies the while law for
refinement.  From step $(3)$ to $(10)$, we refine the inside part of
the while loop and do not display the outer part of the program.

Because now $`i \neq 0$, the list $A$ has at least an element $a$. We
can thus expand the definition of $list$ in step $(3)$. Next, we
assign the value pointed to by $`i + 1$ to $`k$---our first list now
starts at $`k$ and $`i$ points to $[a, `k]$.  Isabelle then struggles
to discharge the generated proof goal automatically.  In this
predicate, the heap is divided in three parts. One needs to prove
first that $`i + 1$ really points to the same value when considering
just the first part of the heap or the entire heap. After that, the
proof is automatic.

Step $(5)$ performs a mutation on the heap, changing the cell $`i + 1$
to $`j$, consequently $`i$ points now to $[a, `j]$.  Because $\ast$ is
commutative, we can strengthen the precondition accordingly in step
$(6)$.  We now work backwards, folding the definition of $list$ in
step $(7)$ and removing the existential of $a$ in step $(8)$. This
step requires again interaction: we need to indicate to Isabelle how
to properly split the heap.  Lastly, to establish the invariant, we
only need to swap the pointers $`j$ to $`i$ and $`i$ to $`k$ in steps
$(9)$ and $(10)$. The resulting algorithm is highlighted in Figure
\ref{fig:ref-reversal}.

% \begin{figure}
% \begin{isabellebody}
% {\isasymlbrakk}
% 	list\ {\isacharbackquote}i\ A$_0${\isacharcomma}\ 
% 	list\ {\isacharbackquote}j\ {\isacharparenleft}rev\ A$_0${\isacharparenright}
% {\isasymrbrakk}
% \isanewline\ \ \ {\isasymsqsubseteq}\isanewline
% {\isacharbackquote}j\ {\isacharcolon}{\isacharequal}\ {\isadigit{0}}{\isacharsemicolon}\isanewline
% $\wwhile$\ {\isacharbackquote}i\ {\isasymnoteq}\ {\isadigit{0}}\ 
% $\ddo$\isanewline
% \ \ {\isacharbackquote}k\ {\isacharcolon}{\isacharequal}\ {\isacharat}{\isacharparenleft}{\isacharbackquote}i\ {\isacharplus}\ {\isadigit{1}}{\isacharparenright}{\isacharsemicolon}\isanewline
% \ \ {\isacharat}{\isacharparenleft}{\isacharbackquote}i\ {\isacharplus}\ {\isadigit{1}}{\isacharparenright}\ {\isacharcolon}{\isacharequal}\ {\isacharbackquote}j{\isacharsemicolon}\isanewline
% \ \ {\isacharbackquote}j\ {\isacharcolon}{\isacharequal}\ {\isacharbackquote}i{\isacharsemicolon}\isanewline
% \ \ {\isacharbackquote}i\ {\isacharcolon}{\isacharequal}\ {\isacharbackquote}k\isanewline
% $\ood$\ 
% \end{isabellebody}
% \caption{\emph{In situ} reversal list algorithm}
% \label{fig:reversal}
% \end{figure}

Using our tool we have also post-hoc verified this algorithm with
separation logic in two different ways. The first one, previously
taken by Weber~\cite{Weber04}, uses Reynolds' list predicate, as we
have used it in the above refinement proof. The second one follows
Nipkow in using separating conjunction in the pre- and postcondition,
but not in the definition of the list predicate. Since our approach is
modular with respect to the underlying data model, it was
straightforward to replay all the steps of Nipkow's proof in our
setting.  The degree of proof automation with our tool is comparable
for both proofs.

Interestingly, however, none of the list reversal proofs have used the
frame rule or its refinement counterpart. We have therefore tested
this rule separately on a small example, where a verification without
the frame rule would be difficult. The following Isabelle code fragment
shows the Hoare triple used for verification.

\begin{center}
\begin{isabellebody}
{\isachardoublequoteopen}{\isasymturnstile}\{\ x\ {\isacharbrackleft}{\isasymmapsto}{\isacharbrackright}\ {\isacharbrackleft}-{\isacharcomma}\ j{\isacharbrackright}\  \ {\isacharasterisk}\  \ list\ j\ as\  \}\ {\isacharat}x\ {\isacharcolon}{\isacharequal}\ a\ \{ x\ {\isacharbrackleft}{\isasymmapsto}{\isacharbrackright}\ {\isacharbrackleft}a{\isacharcomma}\ j{\isacharbrackright}\  \ {\isacharasterisk}\ \ list\ j\ as\  \} {\isachardoublequoteclose}
\end{isabellebody}
\end{center}
Calling the \emph{hoare} tactic for verification condition generation
was sufficient for proving the correctness of this simple example
automatically.  Internally, the frame and the mutation rule have been
applied.
The Isabelle code for all these proofs is available online.

In sum, our approach supports the program construction and
verification of pointer-based programs with separation logic, but more
case studies need to be performed to assess the performance of our
tool. In the future, a Sledgehammer-style integration of optimised
provers and solvers for the data level seems desirable for increasing
the general degree of automation.

%%%%%%%%%%%%%%%%%%%%%%%%%%%%%%%%%%%%%%%%%%%%%%%%%%%%%%%%%%%%%%%%%%%%%%%

\section{Conclusion}\label{S:conclusion}

A principled approach to the design of program verification and
construction tools for separation logic with the Isabelle theorem
proving environment has been presented. This approach has been used
previously for implementing tools for the construction and
verification of simple while programs~\cite{sefm2014} and rely-guarantee
based concurrent programs~\cite{fm2014}. It aims at a clean separation of
concern between the control flow and the data flow of programs and
focusses on developing a lightweight algebraic layer from which
verification conditions or transformation and refinement laws can be
developed by simple equational reasoning. In the case of while
programs, this layer is provided by Kleene algebras with
tests; in the rely-guarantee case, new algebraic
foundations based on concurrent Kleene algebras were required.

Our approach to separation logic uses a conceptual reconstruction of
separation logic beyond a mere implementation as well, which forms a
contribution in its own right. Though strongly inspired by abstract
separation logic~\cite{COY07} and the logic of bunched
implications~\cite{OHearnP99}, we aim at a different combination of
simplicity and mathematical abstraction. In contrast to the logic of
bunched implication, we use power series instead of categories, and in
contrast to abstract separation logic we use predicate transformers
instead of state transformers. These design choices allow us to use
power series, quantales and generic lifting constructions throughout
the approach, which leads to a small and highly automated Isabelle
implementation. The main contribution of this approach is probably the
view on separating conjunction as a notion of convolution and language
product over resources.

Our tool prototype has so far allowed us to verify some simple
pointer-based program with a relatively high degree of automation. It
is certainly a useful basis for educational and research purposes, but
extension and optimisation beyond the mere proof of concept are
desirable. This includes the consideration of error
states~\cite{COY07} or of the $\mathbf{cons}$ and $\mathbf{dispose}$
operations, the development of more sophisticated proof tactics, and
the integration of tools for automatic data-level reasoning in
Sledgehammer style.

Other opportunities for future work lie in the consolidation with
previous approaches to predicate transformers in
Isabelle~\cite{Preoteasa11}, in a further abstraction of the control
flow layer by defining modal Kleene algebras over assertion
quantales~\cite{MoellerStruth} for which some Isabelle infrastructure
already exists~\cite{GuttmannSW11}, in a combination with our
rely-guarantee tool into RGSep-style tools for concurrency
verification~\cite{Vafeiadis}, and in the exploration of the language
connection of separating conjunction in terms of representability and
decidability results.

\paragraph*{Acknowledgements.}  We are grateful for support by EPSRC
grant EP/J003727/1 and the CNPq. The third author would like to thank
Tony Hoare, Peter O'Hearn and Matthew Parkinson for discussions on
separation logic.

%%%%%%%%%%%%%%%%%%%%%%%%%%%%%%%%%%%%%%%%%%%%%%%%%%%%%%%%%%%%%

\bibliographystyle{plain}
\bibliography{sepalg}

\end{document}